\documentclass[11pt, a4paper]{article}%
\usepackage[margin=2.5cm]{geometry}
\usepackage{authblk}
\usepackage{graphicx}
\usepackage{url}
\usepackage{amssymb,mathtools,comment,amsthm,amsmath}
\usepackage{appendix}
\usepackage{color}
\usepackage%[disable]
{todonotes}

\usepackage{multirow}
\usepackage[nocompress]{cite}

\usepackage[colorlinks=true]{hyperref}
\hypersetup{ citecolor = {blue} }
\usepackage{bookmark} % remove some warnings by hyperref
\usepackage{xspace}
\usepackage{bm}
\usepackage[capitalize]{cleveref}
\usepackage{enumerate}

\usepackage{array}

\usepackage[boxed]{algorithm2e}	

\allowdisplaybreaks

\newtheorem{theorem}{Theorem}
\newtheorem{lemma}[theorem]{Lemma}

\newtheorem*{theorem*}{Theorem}

\theoremstyle{remark}

\newcommand{\bigO}{\mathcal{O}}
\newcommand{\polylog}{\mathrm{polylog}}

\newcommand{\Prob}[2]{\mathbb{P}_{#1}\left(#2\right)}
\newcommand{\Expec}[2]{\mathbb{E}_{#1}\left[#2\right]}

\newcommand{\hit}{\mathrm{\tau}} 	%% Hitting time
\newcommand{\p}{p}					%% Transition probabilities
\newcommand{\q}{q}					%% Transition probabilities
\newcommand{\pr}{r}					%% Probability of remaining in current state

\newcommand{\uf}{g}					%% Update function
\renewcommand{\leq}{\leqslant}
\renewcommand{\le}{\leqslant}
\renewcommand{\geq}{\geqslant}
\renewcommand{\geq}{\geqslant}
\renewcommand{\ge}{\geqslant}

\newcommand{\andy}[1]{\todo[inline,color=green]{#1}}
\newcommand{\lb}[1]{\todo[inline,color=yellow]{#1}}
\definecolor{myBlue}{rgb}{0.7, 0.85, 1.}

\definecolor{myRed}{rgb}{1., 0.47, 0.37}
\newcommand{\amos}[1]{\todo[inline,color=myRed]{#1}}

\newcommand{\cnt}{{\tt previousSample}}
\newcommand{\clk}{{\tt countdown}}
\usepackage[utf8]{inputenc}

\title{On the Role of Memory in Robust Opinion Dynamics}
\author{Luca Becchetti, Andrea Clementi, Amos Korman, Francesco Pasquale, Luca Trevisan and Robin Vacus}
\date{December 2022}

\begin{document}

\maketitle

\begin{abstract}
We investigate opinion dynamics in a fully-connected system, consisting of $n$ identical and anonymous agents, where one of the opinions (which is called {\em correct}) represents a piece of information to disseminate. 
In more detail, one \emph{source} agent initially holds the correct opinion and remains with this opinion throughout the execution. The goal for non-source agents is to quickly agree on this correct opinion, and do that robustly, i.e., from any initial configuration. The system evolves in rounds. In each round, one agent chosen uniformly at random is {\em activated}: unless it is the source, the agent pulls the opinions of $\ell$ random agents and then updates its opinion according to some rule. 
We consider a restricted setting, in which agents have no memory and they only revise their opinions on the basis of those of the agents they currently sample. 
%Motivated by biological scenarios, previous works on the problem focused on minimizing the communication capacity
As restricted as it is, this setting encompasses very popular opinion dynamics, such as the \emph{voter model} and \emph{best-of-$k$ majority} rules. 

Qualitatively speaking, we show that lack of memory prevents efficient 
convergence. Specifically, we prove that no dynamics can achieve 
correct convergence in an expected number of steps that is 
sub-quadratic in $n$, even under a strong version of the model in which 
activated agents have complete access to the current configuration of the entire 
system, i.e., the case $\ell=n$. Conversely, we prove that the simple 
voter model (in which $\ell=1$) correctly solves the problem, while almost matching the aforementioned lower bound. 
%At the same time we show that the simple voter model almost achieves the above bound asymptotically, whereas more sophisticated heuristics, such as best-of-$k$ majority dynamics, do not even ensure convergence to the correct opinion. 
%Finally, experimental evidence strongly suggests that dynamics making use of a modest amount of memory can achieve consensus in an almost linear number of steps, corresponding to an exponential gap in the average number of activations per agent. 

These results suggest that, in contrast to symmetric consensus problems (that do not involve a notion of correct opinion), fast convergence on the correct opinion using stochastic opinion dynamics may indeed require the use of memory. This insight may reflect on natural information dissemination processes that rely on a few knowledgeable individuals.

\end{abstract}

\section{Introduction} \label{sec:intro}
\iffalse
Here put a nice Intro please!

Why convergence time in MAS is important:

Typical questions studied are the convergence properties of the
opinion dynamics: Is convergence to stable states guaranteed and
if yes, what are upper and lower bounds on the convergence time?
Guaranteed convergence is essential since otherwise the predictive
power of the model is severely limited. Moreover, studying the
convergence time of opinion dynamics is crucially important. In
general, the analysis of stable states is significantly more meaningful if these states are likely to be reached in a reasonable amount of
time, i.e., if quick convergence towards such states is guaranteed.
If systems do not stabilize in a reasonable time, stable states lack justification as a prediction of the system’s behavior.

%COPIED BY AAMAS PAPER
\fi

Identifying the specific algorithm employed by a biological system is extremely challenging. This quest combines empirical evidence, informed guesses, computer simulations, analyses, predictions, and verifications. One of the main difficulties when aiming to pinpoint an algorithm stems from the huge variety of possible algorithms. This is particularly true when multi-agent systems are concerned, which is the case in many biological contexts, and in particular in collective behavior \cite{sumpter2010collective,feinerman2017individual}. To reduce the space of algorithms, the scientific community often restricts attention to simple algorithms, while implicitly assuming that despite the fact that real algorithms may not necessarily be simple to describe, they could still be approximated by simple rules \cite{couzin2005effective,gelblum2015ant,fonio2016locally}. 
%buhl2006disorder
However, even though this restriction reduces the space of algorithms significantly, the number of simple algorithms still remains extremely large. 

Another direction to reduce the parameter space is to identify classes of algorithms that are less likely to be employed in a natural scenario, for example, because they are unable to efficiently handle the challenges induced by this scenario \cite{boczkowski2018limits,feinerman2017ants,bialek2012statistical}. Analyzing the limits of computation under different classes of algorithms and settings has been a main focus in the discipline of theoretical computer science. 
Hence, following the framework of understanding science through the \emph{computational lens} \cite{karp2011understanding}, it appears promising to employ lower-bound techniques from computer science to biologically inspired scenarios, in order to understand which algorithms are less likely to be used, or alternatively, which parameters of the setting are essential for efficient computation \cite{guinard2021intermittent,boczkowski2018limits}. This lower-bound approach may help identify and characterize phenomena that might be harder to uncover using more traditional approaches, e.g., using
%more complex (and often analytically intractable) models or 
simulation-based approaches or even  differential equations techniques. The downside of this approach is that it is limited to analytically tractable settings, which may be too ``clean'' to perfectly capture a realistic setting.

Taking a step in the aforementioned direction, we focus on a basic problem of information dissemination, in which few  individuals have pertinent information about the environment, and other agents wish to learn this information while using constrained and random communication  \cite{aspnes2009introduction,boczkowski2019minimizing,bastide2021self,DBLP:conf/podc/KormanV22}. 
Such information may include, for example, knowledge about a preferred migration route \cite{franks2002information,lindauer1957communication}, the location of a food source \cite{couzin2011uninformed}, or the need to recruit agents for a particular task \cite{razin2013desert}. 
In some species, specific signals are used to broadcast such information, a remarkable example being the waggle-dance of honeybees that facilitates the recruitment of hive members to visit food sources \cite{franks2002information,seeley2003consensus}. 
%Disseminating information efficiently is a key primitive in many natural or artificial systems consisting of agents with limited capabilities that interact in a stochastic fashion to collectively perform some task \cite{aspnes2009introduction,boczkowski2018limits,razin2013desert,couzin2005effective,angluin2008simple}.
In many other biological systems, however, it may be difficult for individuals to distinguish those who have pertinent information from others in the group \cite{couzin2005effective,razin2013desert}. Moreover,
 in multiple biological contexts, animals cannot rely on distinct signals and must obtain information by merely observing the behavior characteristics of other animals (e.g., their position in space, speed, etc.). This weak form of communication, often referred to as {\em passive communication} \cite{wilkinson1992information}, does not even require animals to deliberately send communication signals \cite{cvikel2015bats,giraldeau2018social}. A key theoretical question is 
identifying minimal computational resources that are 
necessary for information to be disseminated efficiently using passive communication.

Here, following the work in \cite{DBLP:conf/podc/KormanV22}, we consider an idealized model, that is inspired by the following scenario. \\

\noindent{\em Animals by the pond.} Imagine an ensemble of $n$ animals gather around a pond to drink water from it. Assume that one side of the pond, either the northern or the southern side, is preferable (e.g., because the risk of having predators there is reduced). However, the preferable side is known to a few animals only. These informed animals
will therefore remain on the preferable side of the pond. The rest of the group members  would like to learn which side of the pond is preferable, but they are unable to identify which animals are knowledgeable. What they are able to do instead, is to scan the pond and estimate the number of animals on each side of it, and then, according to some rule, move from side to side. Roughly speaking, the main result in \cite{DBLP:conf/podc/KormanV22} is that there exists a rule that allows all animals to converge on the preferable side relatively quickly, despite initially being spread arbitrarily in the pond. The suggested rule essentially says that each agent compares its current sample of the number of agents on each side with the sample obtained in the previous round; If the agent sees that more animals are on the northern (respectively, southern)  side now than they were in the previous sample, then it moves to the northern (respectively, southern) side. 

Within the framework described above, we ask whether knowing anything about the previous samples is really necessary, or whether fast convergence can occur by considering the current sample alone. Roughly speaking, we show that indeed it is not possible to converge fast on the correct opinion without remembering information from previous samples. Next, we describe the model and results in a more formal manner.

%We show that lack of memory has a severe impact on how quickly a system of anonymous and identical agents can converge to the correct opinion. 
\paragraph{Problem definition.}
We consider $n$ agents, each of which holds an {\em opinion} in $\{0,1,\ldots,k\}$, for some fixed integer $k$. One of these opinions is called {\em correct}. One \emph{source} agent\footnote{All results we present 
seamlessly extend (up to constants) to a constant number of source 
agents.} knows which opinion is correct, and hence holds this opinion throughout the execution. The goal of non-source agents is to converge on the correct opinion as fast as possible, from any initial configuration. Specifically, the process proceeds in discrete {\em rounds}. In each round, one
%non-source
agent is sampled uniformly at random (u.a.r) to be {\em activated}.
%initially holds a \emph{correct} opinion, representing the piece of information that is to be disseminated. In 
%the remainder, we use labels $1,\ldots , k$ for the opinions (provided 
%we initially have $k$ different opinions) and we 
%assume that $1$ is the {\em correct} opinion. 
%Starting from an initial configuration in which all agents but the source are in any arbitrary   configuration, the system evolves in rounds. In each round, one agent selected uniformly at random is {\em activated}. 
The activated agent is then given access to the opinions of $\ell$ agents, sampled u.a.r (with replacement\footnote{All the results directly hold also if the sampling is without replacement.}) from the multiset of all the opinions in the population (including the source agent, and the sampling agent itself), for some prescribed integer $\ell$ called {\em sample size}.
If it is not a source,
the agent then revises its current opinion using a decision rule, which defines the {\em dynamics}, and  which  is used by all non-source agents. We restrict attention to dynamics  that are not allowed to switch to opinions that are not contained in the samples they see.
A dynamics is called {\em memoryless} if the corresponding decision rule only depends  on the opinions contained in the current sample and on the opinion of the agent taking the decision. Note that the classical \emph{voter model} and \emph{majority} dynamics are memoryless.

\subsection{Our results}
In Section \ref{sec:lower}, we prove  that every memoryless dynamics 
must have expected running time $\Omega(n^2)$ for every constant number 
of opinions. A bit surprisingly, our analysis holds even under a stronger model in which, in 
every round, the activated agent has access to the current opinions of 
\emph{all} agents in the system.
%, i.e, its sample size is $\ell=n$. 

For comparison, in \emph{symmetric consensus}\footnote{In the remainder, by 
\emph{symmetric consensus} we mean the standard setting in which agents are 
required to eventually achieve consensus on \emph{any} of the opinions that are initially 
present in the system.} 
%convergence to \emph{any} of the opinions that are present in the initial 
%configuration 
convergence is achieved in $\bigO(n\log n)$ rounds with high 
probability, for a large class of majority-like dynamics and using 
samples of constant size \cite{schoenebeck2018consensus}. We thus have an exponential gap 
between the two settings, in terms of the average number of activations 
per agent.\footnote{This measure is often 
referred to as the \emph{parallel time} in distributed computing 
literature \cite{10.1016/j.ipl.2022.106314}.}
%We thus have an exponential gap 
%in terms of the average number of activations per agent\footnote{This measure is often 
%referred to as the \emph{parallel time} in distributed computing 
%literature \cite{10.1016/j.ipl.2022.106314}.} that are necessary to achieve consensus 
%(to the correct opinion in our case, to any opinion in the standard setting). 

We further show that our lower bound is essentially tight. Interestingly, %while majority-like dynamics might not even 
%afford convergence to the right opinion, 
we prove that the standard voter model 
achieves almost optimal performance, despite using samples of size $\ell=1$. Specifically,
 in Section \ref{sec:upper}, we 
 prove that the voter model converges to the correct opinion within 
$\bigO(n^2\log n)$ rounds in expectation and $\bigO(n^2\log^2n)$ rounds with high probability.
%Since the voter model corresponds to $\ell = 1$, 
This result and the lower bound of Section \ref{sec:lower} together
suggest that sample size cannot be a key ingredient in 
achieving fast consensus to the correct opinion after all.

Finally, we argue that allowing agents to use a relatively small amount 
of memory can drastically decrease convergence time. As mentioned 
earlier in the introduction, this result has been formally proved    in~\cite{DBLP:conf/podc/KormanV22} in the \textit{parallel} setting, 
 where at every round, all agents are activated simultaneously. We 
 devise a suitable adaptation of the algorithm proposed in~\cite{DBLP:conf/podc/KormanV22} to work in the sequential, 
 random activation model that is considered in this paper. This  adaptation   uses  samples of size $\ell=\Theta(\log n)$ and $\Theta(\log 
 \log n)$ bits of local memory. Empirical evidence discussed in Section 
 \ref{sec:simulations_short} 
suggests that its convergence time might be compatible with $n \log^{O(1)} n$. In terms of parallel time, this would imply  an exponential gap between this case and the memoryless case.

\subsection{Previous work} \label{sec:related}
The problem we consider spans a number of areas of potential interest. 
Disseminating 
information from a small subset of agents to the larger population is a 
key primitive in many biological, social or artificial systems.
Not surprisingly, dynamics/protocols taking up this challenge have been investigated for 
a long time across several communities, often using different 
nomenclatures, so that terms such as ``epidemics'' 
\cite{demers1987epidemic}, ``rumor spreading'' 
\cite{karp2000randomized}
%or ``gossip'' \cite{kempe2003gossip}
may  refer to the same or similar problems depending on context.

The corresponding literature is vast and providing an exhaustive review is 
infeasible here. In the following paragraphs, we discuss previous 
contributions that most closely relate to the present work.

%\lb{The following paragraphs highlight possible areas of related work. 
%Please change and/or reorganize as you deem appropriate.}

\paragraph{Information dissemination in MAS with limited communication.}
Dissemination is especially difficult when communication is limited 
and/or when the environment is noisy or unpredictable.
For this reason, a line of recent work in distributed computing focuses 
on designing robust protocols, which are tolerant to faults and/or 
require minimal assumptions on the communication patterns. 
An effective theoretical framework to address these challenges is that 
of self-stabilization, in which problems related to the 
scenario we consider have been investigated, such as self-stabilizing clock synchronization or 
majority computations  \cite{aspnes2009introduction,ben2008fast,boczkowski2019minimizing}.

In general however, these models make few assumptions about memory 
and/or communication capabilities and they rarely fit the framework of passive communication.
Extending the self-stabilization framework to  
multi-agent systems arising from biological distributed systems 
\cite{giraldeau2018social,angluin2004computation} has been the focus of 
recent work, with interesting preliminary results  
\cite{DBLP:conf/podc/KormanV22} discussed earlier in the introduction.

\paragraph{Opinion dynamics.}
%\lb{This paragraph should probably be relatively short, mostly stating 
%that the topic is vast and pointing to surveys or key references on the 
%topic.}
Opinion dynamics are mathematical models that have been extensively 
used to investigate processes of opinion formation  resulting in stable   
  consensus   and/or clustering equilibria in multi-agent systems 
\cite{BCN20,CHK18,OZ21}. 
One of the most popular opinion dynamics is the voter model, introduced 
to study spatial conflicts between species in biology
and in interacting  particle systems \cite{CS73,HL75}.
The investigation of majority update rules originates
from the study of consensus processes in spin systems
\cite{KL03}. Over the last decades,  several variants of the basic majority
dynamics  have been studied \cite{BCN20,BHKLRS22,DGMSS11,MNT14}. 

\iffalse 
In particular, bounds on the convergence 
time of  majority rules with small (i.e. constant-size) samples in the 
presence of  byzantine and/or faulty agents have been derived  for the 
complete graphs in \cite{DGMSS11} and for special classes of  graphs in 
\cite{}. 
\fi 

The recent past has witnessed increasing interest for biased variants 
of opinion dynamics, modelling multi-agents systems in which agents may 
exhibit a bias towards one or a subset of the opinions, for example 
reflecting the appeal represented by the diffusion of 
an innovative  technology in a social system. This general problem has been investigated 
under a number of models \cite{ABCPR20,BHKLRS22,CMQR21,LGP22}. 
In general, the focus of this line of work is 
different from ours, mostly being on the sometimes complex interplay between bias and 
convergence to an equilibrium, possibly represented by global adoption 
of one of the opinions. In contrast, our focus is on how quickly 
dynamics can converge to the (unknown) correct opinion, 
i.e., how fast a piece of information can be disseminated within a 
system of anonymous and passive agents, that can infer the 
``correct'' opinion only by accessing random samples of the opinions held  
by other agents.\footnote{For reference, it is easy to 
verify that majority or best-of-$k$ majority rules 
\cite{schoenebeck2018consensus} (which have 
frequently been considered in the above literature)
 in general fail to complete the 
dissemination task we consider.}

%The recent past has seen increasing attention biased op has been devoted to the study 
%of majority rules in the presence  of some forms of bias towards one 
%specific ``preferred'' opinion. This setting aims at modelling 
%multi-agents systems in which there is an intrinsic superiority of one 
%alternative over the others, for instance representing the diffusion of 
%an innovative  technology in a social system.  A key issue here is to 
%establish bounds on the convergence time towards the absorbing state 
%where all the agents get the unique preferred opinion. This issue has 
%been addressed in \cite{ABCPR20,BHKLRS22,CMQR21,LGP22}. It is important to 
%emphasize that all the biased opinion models mentioned above 
%significantly depart from the opinion-formation process we consider in 
%this paper: indeed, while in the former there are no seed (i.e. 
%``superior'') agents and all the agents ``know'' the same biased 
%(majority) opinion, in the latter, a subset of (unknown) agents   keeps 
%the ``superior'' opinion since the very beginning (i.e. they behave as 
%zealots), while the others have to ``recover'' this opinion without 
%any knowledge and/or mechanism helping in doing that. A good way to get 
%a clear evidence of the deep difference between the two frameworks 
%above is to verify how any majority rule definitely fails to achieve 
%the correct consensus in our, former model.
%\lb{This second paragraph needs some improvement --Luca}

\paragraph{Consensus in the presence of zealot agents.} A large body of 
work considers opinion dynamics in the presence of zealot agents, i.e., 
agents (generally holding heterogeneous opinions) that
never depart from their initial opinion~\cite{DCN22,Ma15,MTB20,YOASS13} 
and may try to influence the rest of the agent population. In this case, the 
process resulting from a certain dynamics can result in equilibria 
characterized by multiple opinions. The main focus of this body of work is investigating 
the impact of the number of zealots, their positions in the network and 
the topology of the network itself on such equilibria
\cite{FE14,Ma15,MTB20,YOASS13}, especially when the social goal of the 
system may be achieving self-stabilizing ``almost-consensus'' on 
opinions that are different from those supported  by the zealots.
Again, the focus of the present work is considerably different, so that  
previous results on consensus in the presence of zealots do not carry 
over, at least to the best of our knowledge.

\section{Notations and Preliminaries} \label{sec:prelim}
We consider a system consisting of $n$ \emph{anonymous} agents. 
%Time evolves in 
%rounds. Considered any round, each agent holds one from $k$ discrete 
%opinions, i.e., from the set $\{1,\ldots , k\}$ without loss of 
%generality. In each round, i) one agent $u$ selected uniformly at random is 
%\emph{activated}; ii) the activated agent $u$ is presented with the opinions of 
%$\ell$ other agents in the system, sampled uniformly at random (with 
%replacement); iii) agent $u$ revises its opinion based on those 
%contained in the sample. 
We denote by $x_u^{(t)}$ the opinion held by agent $u$ at the end of round $t$, dropping the superscript whenever it is clear from the context. The \emph{configuration} of the 
system at round $t$ is the vector $\mathbf{x^{(t)}}$ with $n$ entries, such that its 
$u$'th entry is $x_u^{(t)}$.

%Initially, one \emph{source} agent has a \emph{correct} opinion 
%(opinion $1$ without loss of generality) and it never changes its opinion 
%throughout the entire process. All other, non-source agents start in some 
%\emph{arbitrary} initial configuration and they are unaware of 
%which of the opinions is the correct one. 
We are interested in dynamics 
that efficiently \emph{disseminate} the correct opinion. I.e., (i) they 
eventually bring the system into the \emph{correct 
configuration} in which all agents share the correct opinion, and 
%\amos{Amos: we should be consistent when we assume 1 is the correct opinion and when we don't. I suggest that we keep this part %general, i.e, without this assumption and use it only when we really need it.}
(ii) they do so in as few rounds as possible. 
For brevity, we sometimes refer to the latter quantity as 
\emph{convergence time}. 
If $T$ is the convergence time of an execution, we denote 
by $T/n$ the average number of activations per agent, a measure often referred to 
as \emph{parallel time} in the distributed computing literature 
\cite{10.1016/j.ipl.2022.106314}. 
%\noindent\textbf{Remark.} 
For ease of exposition, in the remainder we assume that opinions are binary 
(i.e., they belong to $\{1, 0\}$).
%we consider that the work of opinions is binary, i.e., that opinions are taken from the set $\{1, 0\}$. 
We remark the following: (i) the lower bound on the convergence time  given in Section \ref{sec:lower} already applies by restricting attention to the binary case, and, 
(ii) it is easy to extend the analysis of the voter model given in 
Section \ref{sec:upper} to the general case of $k$ opinions using 
standard arguments. These are briefly summarized in Subection 
\ref{subse:apx_multiple}  for the sake of completeness.

%\andy{I would propose ''convergence time'': it is better to talk about bit-dissemination and not confuse people with other form of consensus.}

\paragraph{Memoryless dynamics.} 
We consider dynamics in which, beyond being anonymous, non-source agents are 
memoryless and identical. 
%Moreover, we require the system to be in a  
%\emph{valid} configuration at any point of the execution, 
%i.e., a configuration containing only opinions that appeared in the initial configuration. 
%\amos{Amos: This validity seems to be out of place. Before, we already said that we look at the binary case. Perhaps you want to move it there?}
%\andy{I agree: I would prefer to not add further notions. It suffices to recall that:  We restrict attention to
%dynamics that are not allowed to switch to opinions that
%are not contained in the samples they see. Then, we can fix adversarially the initial binary conf....from that, the system cannot move to illegal conf.  }
We capture these and the general requirements outlined in Section 
\ref{sec:intro} by the following 
decision rule, describing the behavior of agent $u$ 
%\amos{Amos: let's be consistent with the terminology. I suggest we keep the term agent} 
%activated in any generic round of the execution: 
\begin{enumerate}
	\item $u$ is presented with a uniform 
	sample $S$ of  size $\ell$;
	\item $u$ adopts opinion $1$ with probability $g_{x_u}(|S|)$, 
	where $|S|$ denotes the number of $1$'s in sample $S$. 
\end{enumerate}
Here, $\uf_{x_u}:\{0,\ldots , \ell\}\rightarrow [0, 1]$ is a function  that 
assigns a probability value to the number of ones that appear in $S$. In 
particular, $\uf_{x_u}$ assigns probability zero to opinions with no 
support in the sample, i.e., $\uf_{x_u}(0) = 0$ and $\uf_{x_u}(\ell) = 1$.\footnote{In general, dynamics not meeting this constraint 
cannot enforce consensus.} Note that, in principle, $\uf_{x_u}$ may 
depend on the current opinion of agent $u$.
%(Note, in the binary case, we have just two functions: $\uf_{0}$ and $\uf_{1}$.)

%\andy{So, in the binary case, we have in fact two functions: $\uf_{0}$ and $\uf_{1}$, right...? is this clear in the text?}

The class of 
dynamics described by the general rule above strictly %\amos{Amos: Isn't this class correspond to all memoryless dynamics? I believe it does. If so, we should explicitly say so (plus the text above is confusing).}
%\lb{I am not sure this is the case. For example, the standard majority dynamics is not included, unless one takes $\ell\rightarrow\infty$, which is something I would prefer not to consider. Also: which part of the text is confusing?}
includes all memoryless algorithms that are based on random samples of  fixed size including the popular dynamics, such as the voter model and a large 
class of quasi-majority dynamics 
\cite{liggett2012interacting,schoenebeck2018consensus,BCN20}.

\paragraph{Markov chains.}
In the remainder, we consider discrete time, discrete space Markov chains, whose state 
space is represented by an 
integer interval $\chi = \{z, z + 1,\ldots , n\}$, for suitable $z\ge 1$ and 
$n > z$, without loss of generality (the reason for 
this labeling of the states will be clear in the next sections).
Let $X_t$ be the random 
variable that represents the state of the chain at round $t\ge 0$. The \emph{hitting time} \cite[Section 1]{levin2017markov} of state $x\in S$ is the first time the chain 
is in state $x$, namely:
\[
	\hit_x = \min\{t\ge 0: X_t = x\}.\footnote{Note that the hitting 
	time in general depends on the initial state. Following 
	\cite{levin2017markov}, we specify it when needed.}
\]
%\robin{Consider replacing $\min$ by $\inf$ ?}

A basic ingredient used in this paper 
is describing the dynamics we consider in terms of suitable 
\emph{birth-death} chains, in which the only possible transitions from 
a given state $i\ge z$ are to the states $i$, $i + 1$ (if $i\le n - 1$) and 
$i - 1$ (if $i\ge z + 1$). In the remainder, we denote by 
$\p_i$ and $\q_i$ respectively the probability of moving to 
$i + 1$ and the probability of moving to $i - 1$ when the chain is in 
state $i$. Note that $\p_n = 0$ and $\q_z = 0$. Finally, $\pr_i = 1 - 
\p_i - \q_i$ denotes the probability that, when in state $i$, the chain 
remains in that state in the next step.

\paragraph{A birth-death chain for memoryless dynamics.} 
%\amos{Amos: say that here 1 is assumed to be the correct opinion for the chain here}
The global behaviour of a system with $z$ source agents holding opinion (wlog) $1$  and in which 
all other agents revise their opinions according 
to the general dynamics described earlier when activated, is completely described by a 
birth-death chain $\mathcal{C}_1$ with state space $\{z, ... , n\}$ and the 
following transition probabilities, for $i = z,\ldots n - 1$:
\begin{align}\label{eq:p_i}
    &\p_i = \Prob{}{X_{t+1} = i + 1 \mid X_t = i}\nonumber \\
	& = \frac{n-i}{n}\sum_{s = 
	0}^{\ell}\uf_0(s)\Prob{}{|S| = s \mid X_t = i}\nonumber \\
    & = \frac{n-i}{n}\Expec{i}{\uf_0(|S|)},
\end{align}
where $X_t$ is simply the number of agents holding opinion $1$ at the end of round $t$ and where, following the notation of \cite{levin2017markov}, for a random variable 
$V$ defined over some Markov chain $\mathcal{C}$, we denote by 
$\Expec{i}{V}$ the expectation of $V$ when $\mathcal{C}$ starts in 
state $i$.
%\robin{$\ell$ is not defined in the last equality, so I would write $\Expec{}{g(|S|)}$ instead. Also I would keep the conditioning on $X_t = i$. With these changes, the equation would become
%\begin{align*}
    %&\p_i = \Prob{}{X_{t+1} = i + 1 \mid X_t = i} \\
	%& = \frac{n-i}{n}\sum_{\ell = 
	%0}^k\uf(\ell)\Prob{}{|S| = \ell \mid X_t = i} \\
    %& = \frac{n-i}{n}\Expec{i}{\uf(|S|)}.
%\end{align*}
%}
%\robin{With the aforementioned generalization, this would simply become $\frac{n-i}{n} \cdot \Expec{i}{\uf_0(|S|)}$.}
Eq. \eqref{eq:p_i} follows from the 
law of total probability applied to the possible values for $|S|$ and 
observing that (a) the transition $i \rightarrow i + 
1$ can only occur if an agent holding opinion $0$ is selected for update, which 
happens with probability $(n - i)/n$, and (b) if such an agent observes 
$s$ agent with opinion $1$ in its sample, it will adopt that opinion 
with probability $\uf_0(s)$.  
Likewise, for $i = z + 1,\ldots , n - 1$: 
\begin{equation} \label{eq:q_i}
	\q_i = \Prob{}{X_{t+1} = i - 1 \mid X_t = i} = \frac{i-z}{n}(1 - \Expec{i}{\uf_1(|S|)}),
\end{equation}
%\begin{align}\label{eq:q_i}
%	&\q_i = \Prob{}{X_{t+1} = i - 1 \mid X_t = i}\nonumber\\
%	& = \frac{i - z}{n}\sum_{s = 
%	0}^{\ell}(1 - \uf_1(s))\Prob{}{|S| = s \mid X_t = i}\nonumber \\ 
%	&= \frac{i-z}{n}(1 - \Expec{i}{\uf_1(|S|)}),
%\end{align}
%\robin{With the aforementioned generalization, this would simply become $\frac{i-z}{n} (1-\Expec{i}{\uf_1(|S|)})$.}
with the only caveat that, differently from the previous case, the 
transition $i + 1\rightarrow i$ can only occur if an agent with opinion 
$1$ is selected for update and \emph{this agent is not a source}.
%\amos{Amos: node should be agent, seed should be source}
For this chain, in addition to $\p_n = 0$ and 
$\q_z = 0$ we also have $\q_n = 0$, which follows since $\uf_1(\ell) = 1$. 

We finally note the following (obvious) connections between $\mathcal{C}_1$ and any specific opinion dynamics $P$: (i) the specific birth-death chain for $P$ is obtained from $\mathcal{C}_1$ by specifying the corresponding $\uf_0$ and $\uf_1$ in Eqs. \eqref{eq:p_i} and \eqref{eq:q_i}  above; and (ii) the expected convergence time of $P$ starting in a configuration with $i\ge z$  agents holding opinion 1 is simply $\Expec{i}{\hit_n}$. 
%\robin{What is $k$ here? Also the subscript is missing.}

%\paragraph{The modified chain $\mathcal{C}_2$.}
%\robin{We may consider moving this paragraph to Section~\ref{sec:upper} since $\mathcal{C}_2$ is not used in Section~\ref{sec:lower}.}
%\lb{Yes, this could be a good idea.}
%It should be noted that $\mathcal{C}_1$ has one absorbing state (the 
%state $n$ corresponding to consensus), hence 
%it cannot be reversible. However, we are interested in $\hit_n$, the 
%number of steps to reach state $n$. To characterize $\hit_n$, we 
%consider the chain $\mathcal{C}_2$, with transition probabilities 
%$\p_i$ and $\q_i$ for $i = z,\ldots, n-1$, and with $\p_n = 0$ (as in 
%$\mathcal{C}_1$) but $\q_n = 1$.\footnote{Setting $\q_n = 1$ is only 
%for the sake of simplicity, any positive value will do.} Obviously, for 
%any initial state $i\le n - 1$, $\hit_n$ has exactly the same 
%distribution in $\mathcal{C}_1$ and $\mathcal{C}_2$. For this reason, 
%in the remainder we consider the chain $\mathcal{C}_2$, unless 
%otherwise stated.

\section{Lower Bound} \label{sec:lower}

In this section, we prove a lower bound on the convergence time of memoryless dynamics.
We show that this negative result holds in a very-strong sense: any dynamics must take $\Omega(n^2)$ expected time even if the agents have full knowledge of the current system configuration.
%This fact, together with the result in the next section, shows that the Voter Model is optimal (up to a $\log n$ factor) in the %class of memoryless dynamics, despite its extremely small sample size of 1.

As mentioned in the previous section, we restrict the analysis to the case of two opinions, namely 0 and 1, w.l.o.g. 
%This is without loss of generality since an opinion may only be adopted if it has support in the sample -- and, equivalently here, if it already exists in the population. \amos{Amos: I don't understand the relevance of the last sentence to the assumption. To me, this assumption in the context of lower bounds does not need a justification. (In the particular case above, I don't understand the justification you wrote)}
%\lb{Agree. We have stated this already in the problem definition.}
%\robin{I'm not sure that I agree. Imagine a protocol that, when seeing opinion 0 and 1 in the sample, adopts opinion 2 (for instance to simulate the undecided state strategy). Such protocol can enforce consensus, and their analysis \textit{cannot} be reduced to the binary case. That is why I think the constraint that agents cannot use opinions they don't see is important here. But maybe this is clear enough and does not need explanation?}
To account for the fact that agents have access to the exact configuration of the system, we slightly modify the 
notation introduced in Section~\ref{sec:prelim}, so that here $\uf_{x_u}:\{0,\ldots,n\} \rightarrow [0,1]$
assigns a probability to the number of ones that appear in the population, rather than in a random sample of size~$\ell$.
%\amos{Amos: We are considering sampling without replacement (btw, we should say that explicitly). So, this case corresponds to $\ell=n$. The footnote should be removed.}
%\robin{Actually, I really think we are sampling \textit{with} replacement... We must agree and then clarify this in the problem definition.}
%\lb{Added, emphasized, in problem definition. Once should suffice I hope}
%\footnote{One may think of this as the limit case $\ell \rightarrow +\infty$.}.
Before we prove our main result, we need the following technical results. 

%Their proofs are deferred to Section~\ref{app:lower} of the supplementary materials.
\begin{lemma} \label{lem:lower_bound_sum_inverse}
	For every~$N \in \mathbb{N}$, for every $x \in \mathbb{R}^N$ s.t. for every~$i \in \{1,\ldots,N\}$, $x_i > 0$, we have either
    $\sum_{i=1}^N x_i \geq N$ or $\sum_{i=1}^N \frac{1}{x_i} \geq N$.
    %\begin{equation*}
    %	\sum_{i=1}^N x_i \geq N \quad \text{ or } \quad \sum_{i=1}^N \frac{1}{x_i} \geq N.
    %\end{equation*}
\end{lemma}

\begin{proof} %[Proof of Lemma~\ref{lem:lower_bound_sum_inverse}]
    Consider the case that $\sum_{i=1}^n x_i \leq N$. 
    Using the inequality of arithmetic and geometric means, we can write
    \begin{equation*}
    	1 \geq \frac{1}{N} \sum_{i=1}^N x_i \geq \left(\prod_{i=1}^N x_i \right)^{\frac{1}{N}}.
    \end{equation*}
    Therefore,
    \begin{equation*}
    	1 \leq \left(\prod_{i=1}^N \frac{1}{x_i}\right)^{\frac{1}{N}} \leq \frac{1}{N} \sum_{i=1}^N \frac{1}{x_i},
    \end{equation*}
    which concludes the proof of Lemma~\ref{lem:lower_bound_sum_inverse}.
\end{proof}

%\andy{It seems we have never stated that our expectation for the bounds on T are computed w.r.t. the random selection of the %active agent, the sample,  *and* the randomness of the decision rule (functions g()). Do you think we should do that somewhere?} 
%\amos{Amos: yes, in the definition of convergence time}

\begin{lemma} \label{lem:lower_bound_hit}
    Consider any birth-death chain on~$\{0,\ldots,n\}$. For~$1 \leq i \leq j \leq n$, let $a_i = q_i/p_{i-1}$ and $a(i:j) = \prod_{k=i}^j a_k$. Then, $\Expec{0}{\hit_n} \geq \sum_{1 \leq i < j \leq n} a(i:j)$.
    %\begin{equation*}
    %    \Expec{0}{\hit_n} \geq \sum_{1 \leq i < j \leq n} a(i:j).
    %\end{equation*}
\end{lemma}

\begin{proof} %[Proof of Lemma~\ref{lem:lower_bound_hit}]
    Let~$w_0 = 1$ and for $i \in \{1,\ldots,n\}$, let $w_i = 1/a(1:i)$. The following result is well-known (see, e.g., Eq.~(2.13) in \cite{levin2017markov}). For every~$\ell \in \{1,\ldots,n\}$,
    \begin{equation*} 
        \Expec{\ell-1}{\tau_\ell} = \frac{1}{q_\ell w_\ell} \sum_{i=0}^{\ell-1} w_i.
    \end{equation*}
    Thus,
    \begin{equation*}
    	\Expec{\ell-1}{\tau_\ell} = \frac{1}{q_\ell} \sum_{i=0}^{\ell-1} \frac{a(1:\ell)}{a(1:i)} = \frac{1}{q_\ell} \sum_{i=1}^{\ell} a(i:\ell) \geq \sum_{i=1}^{\ell} a(i:\ell).
    \end{equation*}
    Eventually, we can write
    \begin{equation*}
    	\Expec{0}{\tau_n} = \sum_{\ell = 1}^n \Expec{\ell-1}{\tau_\ell}\geq \sum_{1 \leq i < j \leq n} a(i:j),
    \end{equation*}
    which concludes the proof of Lemma~\ref{lem:lower_bound_hit}.
\end{proof}

\begin{theorem}\label{thm:lowbound}
    Fix~$z \in \mathbb{N}$. In the presence of~$z$ source agents, the expected convergence time of any memoryless dynamics is at least $\Omega(n^2)$, even when each sample contains the complete configuration of the opinions in the system, i.e., the case $\ell=n$.
    % and within $\bigO\left(n^2\log n\log\frac{1}{\delta}\right)$ rounds with probability at least $1 - \delta$, for $0 < \delta < 1$. 
\end{theorem}
\begin{proof}
    Fix $z \in \mathbb{N}$. Let~$n \in \mathbb{N}$, s.t. $n>4z$, and let $P$ be any memoryless dynamics.
    The idea of the proof is to show that the birth-death chain associated with~$P$, obtained from the chain $\mathcal{C}_1$ described in Section \ref{sec:prelim} by specifying $\uf_0$ and $\uf_1$ for the dynamics $P$, cannot be ``fast'' in both directions at the same time.
    %For the sake of neglecting the source agents, %\amos{Amos: neglecting source agents sounds bad. That's the reason I suggested to take $n>4z$.}
    We restrict the analysis to the subset of states $\chi = \{n/4,\ldots,3n/4\}$.
    More precisely, we consider the two following birth-death chains:
    
    \begin{itemize}
        \item $\mathcal{C}$ with state space $\chi$, whose states represent the number of agents with opinion~$1$, and assuming that the source agents hold opinion~$1$.
        \item $\mathcal{C}'$ with state space $\chi$, whose states represent the number of agents with opinion~$0$, and assuming that the source agents hold opinion~$0$.
    \end{itemize}

    Let~$\tau_{3n/4}$ (resp. $\tau'_{3n/4}$) be the hitting time of the state $3n/4$ of chain $\mathcal{C}$ (resp. $\mathcal{C}'$).
    We will show that
    \begin{equation*}
        \max \left( \Expec{n/4}{\tau_{3n/4}} , \Expec{n/4}{\tau'_{3n/4}} \right) = \Omega(n^2).
    \end{equation*}
    Let~$g_0,g_1 : \chi \rightarrow [0,1]$ be the functions describing~$P$ over~$\chi$.
    Following Eqs.~\eqref{eq:p_i} and~\eqref{eq:q_i} in Section~\ref{sec:prelim}, we can derive the transition probabilities for~$\mathcal{C}$ as
    \begin{equation}
        p_i = \frac{n-i}{n} g_0(i), \qquad q_i = \frac{i-z}{n} (1-g_1(i)).
    \end{equation}
    Note that the expectations have been removed as a consequence of agents having ``full knowledge'' of the configuration. 
    Similarly, for~$\mathcal{C}'$, the transition probabilities are
    \begin{equation}
        p_i' = \frac{n-i}{n} (1-g_1(n-i)), \qquad q_i' = \frac{i-z}{n} g_0(n-i).
    \end{equation}
    Following the definition in the statement of Lemma~\ref{lem:lower_bound_hit}, we define~$a_i$ and~$a'_i$ for $\mathcal{C}$ and $\mathcal{C}'$ respectively. We have
    \begin{equation*}
        a_i = \frac{q_i}{p_{i-1}} = \frac{i-z}{n-i+1} \cdot \frac{1-g_1(i)}{g_0(i-1)},
    \end{equation*}
    and
    \begin{equation*}
        a_i' = \frac{q_i'}{p_{i-1}'} = \frac{i-z}{n-i+1} \cdot \frac{g_0(n-i)}{1-g_1(n-i+1)}.
    \end{equation*}
    Observe that we can multiply these quantities by pairs to cancel the factors on the right hand side:
    \begin{equation} \label{eq:a_i_product}
        a_{n-i+1} \cdot a_i' = \frac{i-z}{i} \cdot \frac{n-i+1-z}{n-i+1}.
    \end{equation}
    $(i-z)/i$ is increasing in~$i$, so it is minimized on $\chi$ for~$i = n/4$. Similarly, $(n-i+1-z)/(n-i+1)$ is minimized for~$i = 3n/4$. Hence, we get the following (rough) lower bound from Eq.~\eqref{eq:a_i_product}: for every~$i \in \chi$,
    \begin{equation} \label{eq:a_i_product_lowerbound}
        a_{n-i+1} \cdot a_i' \geq \left(1-\frac{4z}{n}\right)^2.
    \end{equation}
    Following the definition in the statement of Lemma~\ref{lem:lower_bound_hit}, we define~$a(i:j)$ and~$a'(i:j)$ for $\mathcal{C}$ and $\mathcal{C}'$ respectively. From Eq.~\eqref{eq:a_i_product_lowerbound}, we get for any~$i,j \in \chi$ with~$i\leq j$:
    \begin{align*}
        a'(i:j) &\geq \left( 1- \frac{4z}{n} \right)^{2(j-i+1)} \frac{1}{a(n-j+1:n-i+1)} \\
        & \geq \left( 1- \frac{4z}{n} \right)^n \frac{1}{a(n-j+1:n-i+1)}. %\\
        %& \underset{n\to +\infty}{\longrightarrow} \frac{e^{-4z}}{a(n-j+1:n-i+1)}.
    \end{align*}
    Let~$c = c(z) = \exp(-4z)/2$.
    For~$n$ large enough,
    \begin{equation} \label{eq:a_i_prime_lowerbound}
        a'(i:j) \geq \frac{c}{a(n-j+1:n-i+1)}.
    \end{equation}
    Let~$N = n^2/8 + n/4$.
    By Lemma~\ref{lem:lower_bound_sum_inverse}, either
    \begin{equation*}
        \sum_{\substack{i,j \in \chi \\ i<j}} a(i:j) \geq N,
    \end{equation*}
    or (by Eq.~\eqref{eq:a_i_prime_lowerbound})
    \begin{equation*}
        \sum_{\substack{i,j \in \chi \\ i<j}} a'(i:j) \geq c \sum_{\substack{i,j \in \chi \\ i<j}} \frac{1}{a(i:j)} \geq cN.
    \end{equation*}
    By Lemma~\ref{lem:lower_bound_hit}, it implies that either
    \begin{equation*}
        \Expec{n/4}{\tau_{3n/4}} \geq N, \quad \text{ or } \quad \Expec{n/4}{\tau'_{3n/4}} \geq cN.
    \end{equation*}
    In both cases, there exists an initial configuration for which at least~$\Omega(n^2)$ rounds are needed to achieve consensus, which concludes the proof of Theorem~\ref{thm:lowbound}.
\end{proof}

\section{The Voter Model is (Almost) Optimal} \label{sec:upper}
The voter model
is the popular dynamics in which the random agent $v$, activated
at round $t$, pulls another agent $u \in V$ u.a.r. and updates its opinion to the
opinion of $u$.

In this section, we prove that this dynamics achieves consensus within
$\bigO(n^2 \log n)$ rounds in expectation. We prove the result for $z = 1$,
noting that the upper bound can only improve for $z > 1$. Without loss of generality, we assume that 1 is the correct opinion.

\paragraph{The modified chain $\mathcal{C}_2$.}
In principle, we could study convergence of the voter model using the 
chain $\mathcal{C}_1$ introduced in Section \ref{sec:prelim}  and 
used to prove the results of Section \ref{sec:lower}. 
Unfortunately, $\mathcal{C}_1$ has one absorbing state (the 
state $n$ corresponding to consensus), hence 
it is not reversible, so that we cannot leverage 
known properties of reversible birth-death chains \cite[Section
2.5]{levin2017markov} that would simplify the proof. Note however that 
we are interested in $\hit_n$, the 
number of rounds to reach state $n$ under the voter model. To this 
purpose, it is possible to 
consider a second chain $\mathcal{C}_2$ that is almost identical to 
$\mathcal{C}_1$ but reversible. In particular, the transition probabilities 
$\p_i$ and $\q_i$ of $\mathcal{C}_2$ are the same as in $\mathcal{C}_1$, 
for $i = z,\ldots, n-1$. Moreover, we have $\p_n = 0$ (as in 
$\mathcal{C}_1$) but $\q_n = 1$.\footnote{Setting $\q_n = 1$ is only 
for the sake of simplicity, any positive value will do.} Obviously, for 
any initial state $i\le n - 1$, $\hit_n$ has exactly the same 
distribution in $\mathcal{C}_1$ and $\mathcal{C}_2$. For this reason, 
in the remainder of this section we consider the chain $\mathcal{C}_2$, unless 
otherwise stated.

\begin{theorem}\label{thm:uppbound}
For $z = 1$, the voter model achieves consensus to opinion $1$ within $\bigO(n^2
\log n)$ rounds in expectation and within $\bigO\left(n^2\log
n\log\frac{1}{\delta}\right)$ rounds with probability at least $1 - \delta$, for
$0 < \delta < 1$. 
\end{theorem}
%\begin{proof}

\subsection{Proof of Theorem \ref{thm:uppbound}}
We first compute the general expression for $\Expec{z}{\hit_n}$, i.e., the
expected time to reach state $n$ (thus, consensus) in $\mathcal{C}_2$ when the
initial state is $z$, corresponding to the system starting in a state in which
only the source agents hold opinion $1$. We then give a specific upper bound
when $z = 1$. First of all, we recall that, for $z$ source agents we have that
$\Expec{z}{\hit_n} = \sum_{k = z + 1}^n\Expec{k-1}{\hit_k}$. Considering the
general expressions of the $\p_i$'s and $\q_i$'s in~Eq.~\eqref{eq:p_i}
and~Eq.~\eqref{eq:q_i}, we soon observe that for the voter model $g_0 = g_1 =g$,
since the output does not depend on the opinion of the agent, and
$\Expec{}{g(|S|)} = i / n$ whenever the number of agent with opinion $1$ in the
system is $i$. Hence for $\mathcal{C}_2$ we have
\begin{equation}\label{eq:voter_transition_probs}
\begin{aligned}
    p_i & = 
    \left\{
    \begin{array}{cl}
        \frac{(n-i)i}{n^2}, & \text{ for } i = z,\ldots , n-1\\[1mm]
        0, & \text{ for } i = n
    \end{array}
    \right.\\
    q_i & =
    \left\{
    \begin{array}{cl}
        0, & \text{ for } i = z \\[1mm]
        \frac{(n-i)(i-z)}{n^2}, & \text{ for } i = z + 1, \ldots, n - 1\\[1mm]
        1, & \text{ for } i = n.
    \end{array}
    \right.
\end{aligned}
\end{equation}
The proof now proceeds along the following rounds.

\paragraph{General expression for $\Expec{k-1}{\hit_k}$.}
It is not difficult to see that
\begin{equation}\label{eq:one_step_expect}
\Expec{k-1}{\hit_k} = \frac{1}{\q_kw_k}\sum_{j=z}^{k-1}w_j\,,
\end{equation}
where $w_0 = 1$ and $w_k = \prod_{i = z+1}^k\frac{\p_{i-1}}{\q_i}$, for $k = z
+ 1, \ldots, n$.  Indeed, the $w_k$'s satisfy the detailed balanced conditions
$\p_{k-1}w_{k-1} = \q_k w_k$ for $k = z+1,\ldots, n$,
\begin{align*}
\p_{k-1}w_{k-1} 
  & = \p_{k-1}\prod_{i = z+1}^{k-1}\frac{\p_{i-1}}{\q_i} \\
  & = \p_{k-1}\frac{\q_k}{\p_{k-1}}\prod_{i = z+1}^k\frac{\p_{i-1}}{\q_i} 
	= \q_kw_k.
\end{align*}
and~Eq.~\eqref{eq:one_step_expect} follows proceeding like in~\cite[Section
2.5]{levin2017markov}.

\paragraph{Computing $\Expec{k-1}{\hit_k}$ for $\mathcal{C}_2$.}
First of all, considering the expressions of $\p_i$ and $\q_i$
in~Eq.~\eqref{eq:voter_transition_probs}, for $k = z + 1,\ldots , n - 1$ we have
\begin{align*}
	w_k & = \prod_{i=z+1}^{k}\frac{(n-i+1)(i-1)}{(i-z)(n-i)} \\
    & =	\prod_{i=z+1}^{k}\frac{n-i+1}{n-i}\cdot\prod_{i=z+1}^{k}\frac{i-1}{i-z} 
    = \frac{n-z}{n-k} \cdot\prod_{i=z+1}^{k}\frac{i-1}{i-z}.
\end{align*}
Hence
\[
	w_k = 
    \left\{
    \begin{array}{cl}
    \frac{n-z}{n-k}f(k), & \; \mbox{ for } k = z + 1,\ldots , n - 1 \\[1mm]
    \frac{(n-z)(n-1)}{n^2}f(n-1), & \; \mbox{ for } k = n 
    \end{array}
    \right.
\]
where $f(k) = \prod_{i=z+1}^{k}\frac{i-1}{i-z}$.

\paragraph{The case $z = 1$.}
In this case, the formulas above simplify and, for $k = 
z+1,\ldots , n-1$, we have
\[
\Expec{k-1}{\hit_k}
= \frac{n^2}{(k-1)f(k)}\sum_{j=1}^{k-1}\frac{f(j)}{n-j}
= \frac{n^2}{k-1}\sum_{j=1}^{k-1}\frac{1}{n-j},
\]
where the last equality follows from the fact that that $f(z) = f(z+1) = \cdots
= f(k) = 1$, whenever $z = 1$. Moreover, for $k = n$ we have
\begin{align*}
	\Expec{n-1}{\hit_n} & = \frac{1}{\q_nw_n}\sum_{j=1}^{n-1}w_j
    = \left(\frac{n}{n-1}\right)^2 \, \sum_{j=1}^{n-1}\frac{n-1}{n-j} \\
    & = \frac{n}{n-1}H_{n-1} = \bigO(\log n),
\end{align*}
where $H_{k}$ denotes the $k$-th harmonic number. 
Hence, for $z = 1$ we have
\begin{align}\label{eq:expected_ub}
	& \Expec{1}{\hit_n} = \sum_{k=2}^n\Expec{k-1}{\hit_k} \nonumber \\
    & = n^2 \sum_{k=2}^{n-1}\frac{1}{k-1}\sum_{j=1}^{k-1}\frac{1}{n-j} + 
	\bigO(\log n),
\end{align}
where in the second equality we took into account that $\Expec{n-1}{\hit_n} = \bigO(\log n)$. 
Finally, it is easy to see that 
\begin{equation}\label{eq:double_harmonic}
\sum_{k=2}^{n-1}\frac{1}{k-1}\sum_{j=1}^{k-1}\frac{1}{n-j} = \bigO(\log n)
\end{equation}
Indeed, if we split the sum at $\lfloor n/2 \rfloor$, for $k \leqslant \lfloor
n/2 \rfloor$ we have 
\begin{equation}\label{eq:dh_1}
\sum_{k=2}^{\lfloor n/2 \rfloor}\frac{1}{k-1}\sum_{j=1}^{k-1}\frac{1}{n-j}
\leqslant \sum_{k=2}^{\lfloor n/2 \rfloor}\frac{1}{k-1}\sum_{j=1}^{k-1}\frac{2}{n}
= \bigO(1)
\end{equation}
and for $k > \lfloor n/2 \rfloor$ we have
\begin{align}\label{eq:dh_2}
& \sum_{k = \lfloor n/2 \rfloor + 1}^{n-1}\frac{1}{k-1}\sum_{j=1}^{k-1}\frac{1}{n-j}
\leqslant \sum_{k = \lfloor n/2 \rfloor + 1}^{n-1} \frac{2}{n}\sum_{j=0}^{n-1}\frac{1}{n-j} \nonumber \\
& = \sum_{k = \lfloor n/2 \rfloor + 1}^{n-1} \frac{2}{n} H_n = \bigO(\log n)\,.
\end{align}
From~Eqs. \eqref{eq:dh_1} and~\eqref{eq:dh_2} we get~Eq. \eqref{eq:double_harmonic}, and
the first part of the claim follows by using in~Eq. \eqref{eq:expected_ub} the
bound in~Eq. \eqref{eq:double_harmonic}.

To prove the second part of the claim, we use a standard argument. Consider
$\lceil \log \frac{1}{\delta}\rceil$ consecutive time intervals, each consisting
of $s = 2 \lceil\Expec{1}{\hit_n}\rceil = \bigO(n^2 \log n)$ consecutive rounds.
For $i = 1,\ldots , s - 1$, if the chain did not reach state $n$ in any of the
first $i - 1$ intervals, then the probability that the chain does not reach
state $n$ in the $i$-th interval is at most $1/2$ by Markov's inequality. Hence,
the probability that the chain does not reach state $n$ in any of the intervals
is at most $\left(1/2\right)^{\log (1/\delta)} = \delta$.
%\end{proof}

%\paragraph{Remark.} It might seem reasonable to consider a different dynamics (generalizing the voter model), in which an agent samples $\ell$ neighbours uniformly and independently at random and then it adopts opinion $1$ with probability $k/\ell$ if $k$ is the number of agent with that opinion in the sample. It is easy to check that this change has no effect, in the sense that the transition probabilities of the 
%corresponding birth-death chain do not change.  

\subsection{Handling multiple opinions}\label{subse:apx_multiple}
Consider the case in which the set of possible opinions is $\{1,\ldots, 
k\}$ for $k\ge 2$, with $1$ again the correct opinion. We collapse 
opinions $2,\ldots , k$ into one class, i.e., opinion $0$ without loss 
of generality. We then consider the random variable $X_t$, giving the 
number of agents holding opinion $1$ at the end of round $t$. Clearly, 
the configuration in which all agents hold opinion $1$ is the only 
absorbing state under the voter model and 
convergence time is defined as $\min\{t\ge 0: X_t = n\}$. For a generic 
number $i$ of agents holding opinion $1$, we next compute the 
probability $p_i$ of the transition $i\rightarrow i + 1$ (for $i\le n 
- 1$) and the probability $q_i$ of the transition $i\rightarrow i - 1$ 
(for $i\ge z + 1$):
\begin{align*}
    &\p_i = \Prob{}{X_{t+1} = i + 1 \mid X_t = i} = \frac{n-i}{n}\cdot\frac{i}{n},
\end{align*}
where the first factor in the right hand side of the above equality is 
the probability of activating an agent holding an opinion other than 
$1$, while the second factor is the probability that said agent in turn 
copies the opinion of an agent holding the correct opinion. Similarly, we 
have:
\begin{align*}
	&\q_i = \Prob{}{X_{t+1} = i - 1 \mid X_t = i} = \frac{i - 
	z}{n}\cdot\frac{n-i}{n},
\end{align*}
with the first factor in the right hand side the probability of 
sampling a non-source agent holding opinion $1$ and the second factor 
the probability of this agent in turn copying the opinion of an agent 
holding any opinions other than $1$.
%=======
%In Section~\ref{sec:upper} we proved an upper bound $\bigO(n^2 \log n)$ on the
%expected convergence time of the voter model when there are only two opinions.
%However, it is straightforward to generalize that proof to the case of multiple
%opinions by collapsing in one single opinion all the opinions different from
%that of the source agents. Indeed, in the case we start from a configuration in
%which each agent has an opinion taken from a set $K$ with $|K| > 2$ and we run
%the voter model, if we define $X_t$ as the number of agents holding the opinion
%of the source agents at round $t$ then $\{X_t\}_{t \in \mathbb{N}}$ is a Markov
%chain whose transition probabilities from state $i$ to states $i \pm 1$ are
%those defined in~\eqref{eq:voter_transition_probs}.
%>>>>>>> 225123d3a6c1e4a71033b28ede2383be6761120b

The above argument implies that if we are interested in the time to 
converge to the correct opinion, variable $X_t$ is what we are actually 
interested in. On the other hand, it is immediately clear that the 
evolution of $X_t$ is described by the birth-death chain 
$\mathcal{C}_1$ introduced in Section \ref{sec:prelim} (again with $n$ 
as the only absorbing state) or by its reversible counterpart 
$\mathcal{C}_2$. This in turn implies that the 
analysis of Section \ref{sec:upper} seamlessly carries over to the case of 
multiple opinions. 
\section{Faster Dissemination with Memory} \label{sec:simulations_short}

In this section, we give experimental evidence suggesting that dynamics using a modest amount of memory can achieve consensus in an almost linear number of rounds. When compared to memory-less dynamics, this represents an exponential gap (following the results of Section~\ref{sec:lower}).

The dynamics that we use is derived from the algorithm introduced in  \cite{DBLP:conf/podc/KormanV22} and is described in the next Subsection.

\subsection{``Follow the trend'': our candidate approach} The dynamics that we run in the simulations is derived from the algorithm of \cite{DBLP:conf/podc/KormanV22}, and uses a sample size of $\ell = 10 \, \log n$.
Each time an agent is activated, it decrements a countdown by 1. When the countdown reaches 0, the corresponding activation is said to be \textit{busy}.
On a busy activation,
the agent compares the number of opinions equal to 1 that it observes, to the number observed during the last busy activation.
\begin{itemize}
    \item If the current sample contains more 1's, then the agent adopts the opinion 1.
    \item Conversely, if it contains less 1's, then it adopts the opinion 0.
    \item If the current sample contains exactly as many 1's as the previous sample, the agent remains with the same opinion.
\end{itemize}
At the end of a busy activation, the agent resets its countdown to~$\ell$ (equal to the sample size) -- so that there is exactly one busy activation every $\ell$ activations. In addition, the agent memorizes the number of 1's that it observed, for the sake of performing the next busy activation. 

Overall, each agent needs to store two integers between $0$ and~$\ell$ (the countdown and the number of opinions equal to 1 observed during the last busy activation), so the dynamics requires $2\log(\ell) = \Theta(\log \log n)$ bits of memory.

The dynamics is described formally in Algorithm~\ref{algo:sequential_follow_the_trend}.

\begin{algorithm}[!ht]
\SetKwInOut{Input}{Input}
\caption{Follow the Trend \label{algo:sequential_follow_the_trend}}
\DontPrintSemicolon
    {\bf Sample size:} $\ell = 10 \log n$ \;
    {\bf Memory:} $\cnt,\clk \in \{ 0,\ldots,\ell \}$ \;
    \Input{$k$, number of ones in a sample}% \\ $x_u$, opinion of $u$}
    \BlankLine

    \uIf{$\clk = 0$} {
        \;
        \uIf{$k < \cnt$}{
    		$x_u \leftarrow 0$ \;
    	}
    	\ElseIf{$k > \cnt$}{
    	    $x_u \leftarrow 1$ \;
         }\;
         $\cnt \leftarrow k$ \;
         $\clk \leftarrow \ell$ \;
    } \Else {
        $\clk \leftarrow \clk - 1$ \;
    }
\end{algorithm}

\subsection{Experimental results}
Simulations were performed for~$n=2^i$, where $i \in \{3,\ldots,10\}$ for the Voter model, and $i \in \{3,\ldots,17\}$ for Algorithm~\ref{algo:sequential_follow_the_trend}, and repeated $100$ times each.
Every population contains a single source agent ($z=1$).
In self-stabilizing settings, it is not clear what are the worst initial configurations for a given dynamics. Here, we looked at two different ones:
\begin{itemize}
    \item a configuration in which all opinions (including the one of the source agent) are independently and uniformly distributed in~$\{0,1\}$,
    \item a configuration in which the source agent holds opinion~$0$, while all other agents hold opinion~$1$.
\end{itemize}

%\begin{figure}[htbp]
%    \centering
%    \includegraphics[width=0.6\textwidth]{resource/plot2.png}
%    \caption{\em
%    {\bf ``Follow the Trend'' versus the Voter model.}
%    Average convergence time (in parallel rounds) is depicted for different values of~$n$, over 100 iterations each, for $z=1$ source agent.
%    Blue lines with circular markers correspond to our candidate dynamics which is an adaptation of the ``follow the trend'' algorithm from \cite{DBLP:conf/podc/KormanV22}.
%    Orange lines with triangular markers correspond to the Voter Model.
%    Full lines depict initial configurations in which all opinions are drawn uniformly at random from~$\{0,1\}$.
%    Dotted lines depict initial configurations in which all opinions are taken to be different from the source agent.
%    }
%    \label{fig:simulations_memory}
%\end{figure}

We compare it experimentally to the voter model.
Simulations were performed for~$n=2^i$, where $i \in \{3,\ldots,10\}$ for the voter model, and $i \in \{3,\ldots,17\}$ for our candidate dynamics, and repeated $100$ times each.
Results are summed up in Figure~\ref{fig:simulations_memory_short}, in terms of parallel rounds (one parallel round corresponds to~$n$ activations).
They suggest that    convergence of  our candidate dynamics takes  about~$\Theta(\polylog n)$ parallel rounds.
In terms of parallel time, this represents an exponential gap when compared to the lower bound in Theorem \ref{thm:lowbound} established for memoryless dynamics.

\begin{figure}[htbp]
    \centering
    \includegraphics[width=0.6\textwidth]{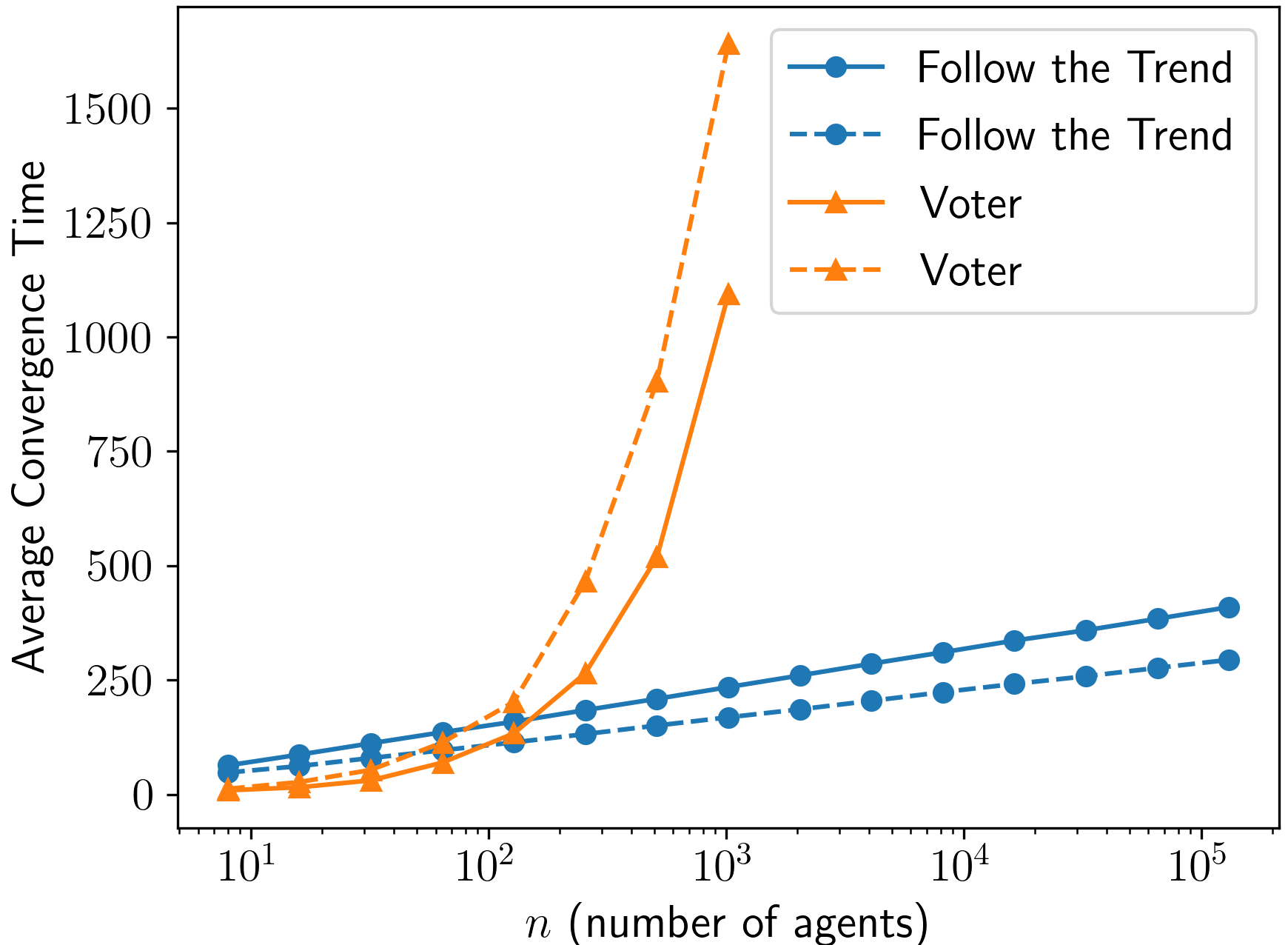}
    \caption{\em
    {\bf ``Follow the Trend'' versus the voter model.}
    Average convergence time (in parallel rounds) is depicted for different values of~$n$, over 100 iterations each, for $z=1$ source agent.
    Blue lines with circular markers correspond to our candidate dynamic which is an adaptation of the ``follow the trend'' algorithm from \cite{DBLP:conf/podc/KormanV22}.
    Orange lines with triangular markers correspond to the voter model.
    Full lines depict initial configurations in which all opinions are drawn uniformly at random from~$\{0,1\}$.
    Dotted lines depict initial configurations in which all opinions are taken to be different from the source agent.
    }
    \label{fig:simulations_memory_short}
\end{figure}

\section{Discussion and Future Work} \label{sec:conc}
This work investigates the role played by memory in multi-agent systems that rely on passive communication and aim to achieve consensus on an opinion held by few ``knowledgable'' individuals \cite{DBLP:conf/podc/KormanV22,couzin2005effective,ayalon2021sequential}. Under the model we consider, we prove that incorporating past observations in the current decision is necessary for achieving fast convergence even if the observations regarding the current opinion configuration are complete. The same lower bound proof can in fact be adapted to any process that is required to alternate the consensus (or semi-consensus) opinion, i.e., to let the population agree (or almost agree) on one opinion, and then let it agree on the other opinion, and so forth. Such oscillation behaviour is fundamental to sequential decision making processes \cite{ayalon2021sequential}.

The ultimate goal of this line of research is to reflect on biological processes and conclude lower bounds on biological parameters. However, despite the generality of our model, more work must be done to obtain concrete biological conclusions. 
%In flocking or schooling processes, for example, aiming to converge on a desirable direction involves opinions (movement directions) that are taken from a continuous range \cite{couzin2005effective}. Moreover, updating rules are not restricted to the initial set of opinions (as assumed in our work), and instead, new opinions may be formed e.g., by averaging directions \cite{couzin2005effective,couzin2011uninformed,korman2022distributed,gelblum2015ant}. Even if the space of opinions is binary, as considered in the decision-making experiment on ants in \cite{ayalon2021sequential}, opinions may still form a larger set (such as the direction of attachment to the carried load in \cite{ayalon2021sequential}), and in order to disqualify this phenomenon, one must readjust the experimental setting. Moreover, in biology, what sometimes appears as ``neutral'' passive communication can, in fact,  encode additional useful information. For example, in the process of recruitment, the speed of ants has been shown to encode information about the certainty of their opinions \cite{razin2013desert,korman2014confidence}. Therefore, again, observing the behavior of individuals may provide access to more information than merely their opinions. 
Conducting an experiment that fully adheres to our model, or refining our results to apply to more realistic settings remains for future work. Candidate experimental settings that appear to be promising include fish schooling \cite{couzin2005effective,couzin2011uninformed}, collective sequential decision-making in ants \cite{ayalon2021sequential}, and recruitment in ants \cite{razin2013desert}. 
If successful, such an outcome would be highly pioneering from a methodological perspective.  Indeed, to the best of our knowledge, a concrete lower bound on a biological parameter that is achieved in an indirect manner via  mathematical considerations has never been obtained. 

\paragraph{\bf Acknowledgement.} The authors would like to thank Yoav Rodeh for very helpful discussions concerning the lower bound proof (Theorem~\ref{thm:lowbound}).

%\amos{Amos: I would like that we add acknowledgment to Yoav Rodeh for very helpful discussions concerning the lower bound proof.}
%\robin{I've already thought about it, but apparently this is not allowed in the submission. We have to wait until we write the camera-ready version.}
%Finally, from a technical perspective, it would be interesting to study the convergence time of the voter-model in general topologies.

%\andy{We should think whether to discuss the following:}

%\begin{itemize}
%    \item The extension of the voter-model analysis (i.e. the upper bound) to general topologies. I am aware that we don't have anything concrete now, but a referee whp will pose this question.
 %   \item possible experimental results on random graphs? the ijcai paper \cite{} which is very similar to us, did it....

 %   \item Analytical results for the setting with memory
    
%\end{itemize}

%\clearpage
\bibliographystyle{unsrt}
\bibliography{zeabib}

%\clearpage
%\appendix

%\input{trunks/Appendix.tex}
%\input{trunks/Simulations.tex}

\end{document}